\documentclass[12pt]{article}

\usepackage[utf8]{inputenc}

\usepackage{theorem}
\newtheorem{fact}{Fact}
\newtheorem{lemma}{Lemma}
\newtheorem{proof}{Proof}
\newtheorem{remark}{Remark}
\newtheorem{theorem}{Theorem}

\usepackage{a4wide}

\usepackage{amsmath}
\usepackage{amssymb}
\usepackage{verbatim}
\usepackage{listings}
\usepackage[ruled,lined,linesnumbered]{algorithm2e}
\usepackage{color}
\usepackage{graphicx}
\usepackage{tabularx}
\newcolumntype{R}{>{\raggedleft\arraybackslash}X}
\usepackage{tikz}
\usepackage{amstext}
\usepackage{appendix}
\usepackage{todonotes}
\usepackage{colortbl}
\usepackage{lscape}
\usepackage{caption}
\usepackage{subcaption}
\usepackage[clock]{ifsym}
\usepackage{stmaryrd}
\usepackage{wrapfig}
\usepackage{latexsym}

\usepackage{xspace}
\usepackage[inline]{enumitem}
\usepackage{color, colortbl}
\definecolor{Gray}{gray}{0.9}
\usepackage{xcolor}
\colorlet{vert}{green!70!black}
\colorlet{rouge}{red!70!black}
\colorlet{orange}{orange!100!black}
\colorlet{bleu}{cyan!80!white!80!black}
\colorlet{gris}{black!10!white}

\newcommand{\green}[1]{\textcolor{vert}{#1}}
\newcommand{\red}[1]{\textcolor{rouge}{#1}}

\usepackage{caption}

\usepackage{wasysym}

\usepackage{array}
\newcolumntype{H}{>{\setbox0=\hbox\bgroup}c<{\egroup}@{}}

\def\qed{}

\usepackage[strings]{underscore}

\SetAlCapSkip{2mm}	

\newcommand{\calR}{\mathcal{R}}

\newcommand{\true}{\mathit{True}}
\newcommand{\false}{\mathit{False}}

\newcommand{\val}{\nu}
\newcommand{\semof}[1]{{\llbracket #1 \rrbracket}}

\newcommand{\langof}{\mathcal{L}}

\newcommand{\bexpr}{\beta}

\def\ie{{i.e.},~}
\def\eg{{e.g.},~}

\def\uppaal{\textsc{Uppaal}\xspace}
\def\imitator{\textsc{Imitator}\xspace}
\def\symrob{\textsc{symrob}\xspace}
\def\phaver{\textsc{PHAver}\xspace}
\def\tar{\textsc{rttar}\xspace}
\def\zthree{\textsc{Z3}\xspace}

\newcommand{\figref}[1]{\figurename~\ref{#1}}
\newcommand{\tabref}[1]{\tablename~\ref{#1}}
\newcommand{\algref}[1]{Algorithm~\ref{#1}}

\def\timed{\tau}

\def\setR{\mathbb{R}}
\def\setQ{\mathbb{Q}}
\def\lang{{\cal L}}
\def\tlang{{\cal TL}}

\def\calU{{\cal U}}
\def\calR{{\cal R}}
\def\calI{\Sigma}

\usepackage{floatflt}

\newcommand{\fut}[1]{\semof{#1}}
\def\unt{\textit{Unt}}
\def\post{\textit{Post}}
\def\enc{\textit{Enc}}
\def\vars{\textit{Vars}}
\def\ia{\textit{ITA}}
\def\test{\textit{Assume}}
\def\deidx{\pi}
\def\emptyset{\varnothing}

\newcommand{\prog}{P}
\newcommand{\states}{Q}
\newcommand{\astate}{q}
\newcommand{\alphabet}{\Sigma}
\newcommand{\lbl}{\alpha}
\newcommand{\accept}{F}
\newcommand{\automata}{\mathcal{A}}
\newcommand{\rvars}{V}
\newcommand{\csystem}{\varphi}
\newcommand{\updt}{\mu}
\newcommand{\rates}{\rho}
\newcommand{\guard}{\gamma}
\newcommand{\vals}{\mathcal{V}}
\newcommand{\tfunc}{\Delta}
\newcommand{\word}{\sigma}

\newcounter{mynote}
\newlength\mynotewidth
\usepackage{tikz}
\usetikzlibrary{petri,calc,automata,arrows,decorations.pathreplacing,decorations.pathmorphing,fit,positioning,patterns,shapes,fadings,backgrounds}
\tikzset{initial text={},
    every state/.style={circle,minimum size=.5cm,draw=blue!50,very thick,fill=blue!20},
    green/.style={rectangle, rounded corners,fill=vert,text=white,thick},
    red/.style={rectangle, rounded corners,fill=rouge,text=white,thick},
    adam/.style={rectangle, rounded corners,color=white,draw=black,thick,fill=black!80!white},
    module/.style={rectangle, rounded corners,color=white,draw=none,thick,fill=bleu},
    node distance=2cm and 3cm,
    error/.style={accepting,double distance=1.5pt,fill=red},
    infeas/.style={draw=red,accepting,double distance=1.5pt,fill=red},
    on grid,
    auto, 
}

\begin{document}

\title{Verification and Parameter Synthesis for Real-Time Programs using Refinement of Trace Abstraction\footnote{A preliminary version of this work appeared in~\cite{rprttar}.}}

\author{Franck Cassez\\
  Department of Computing\\
  Macquarie University\\ 
  Sydney, Australia
\and  Peter Gj{\o}l Jensen \\
 Department of Computer Science\\ Aalborg University\\Denmark
\and Kim Guldstrand Larsen \\
 Department of Computer Science\\ Aalborg University\\Denmark
 }

\maketitle

\begin{abstract}
We address the safety verification and synthesis problems for real-time systems.
We introduce real-time programs that are made of instructions that can perform assignments to discrete and real-valued variables. They are general enough to capture interesting classes of timed systems such as timed automata, stopwatch automata, time(d) Petri nets and hybrid automata.
We propose a semi-algorithm using refinement of trace abstractions to solve both the reachability verification problem and the parameter synthesis problem for real-time programs.
All of the algorithms proposed have been implemented and we have conducted a series of experiments, comparing the performance of our new approach to state-of-the-art
tools in classical reachability, robustness analysis and parameter synthesis for timed systems.
We show that our new method provides solutions to problems which are unsolvable by the current state-of-the-art tools.
\end{abstract}

\section{Introduction} \label{sec:intro}

Model-checking is a widely used formal method to assist in verifying software systems.
A wide range of model-checking techniques and tools are available and there are numerous successful applications in the safety-critical industry and the hardware industry -- in addition the approach is seeing an increasing adoption in the general software engineering community.
The main limitation of this formal verification technique is the so-called \emph{state explosion problem}.
\emph{Abstraction refinement techniques} were introduced to overcome this problem.
The most well-known technique is probably the \emph{Counter Example Guided Abstraction Refinement} (CEGAR) method pioneered by Clarke \emph{et al.}~\cite{Clarke2000}. In this method the state space is abstracted
with predicates on the concrete values of the program variables.
The (counter-example guided) \emph{trace abstraction refinement} (TAR) method was proposed later by Heizmann \emph{et al.}~\cite{traceref,HeizmannHP-cav-13} and is based on abstracting the set of traces of a program rather than the set of states.
These two techniques have been widely used in the context of software verification. Their effectiveness and versatility in verifying \emph{qualitative} (or functional) properties of C programs is reflected in the most recent \emph{Software Verification} competition results~\cite{DBLP:conf/tacas-2019-3}.

\paragraph{\bfseries Analysis of timed systems.}

Reasoning about \emph{quantitative} properties of programs requires
extended modeling features like real-time clocks.
\emph{Timed Automata}~\cite{AlurDill} (TA), introduced by Alur and Dill in 1989, is  a
very popular formalism to model real-time systems with dense-time clocks.
Efficient symbolic model-checking techniques for TA are implemented in the real-time model-checker \uppaal~\cite{uppaal}.
Extending TA, \eg with the ability to stop and resume clocks (stopwatches), leads to undecidability of the reachability problem~\cite{stopwatches,Henzinger199894} which is the basic verification problem.
As a result, semi-algorithms have been designed to verify extended classes of TA \eg \emph{hybrid automata},  and are implemented in a number of dedicated tools~\cite{phaver,spaceex,Henzinger97hytech}.
However, a common difficulty with the analysis of quantitative properties of timed automata and extensions thereof is that specialized data-structures are needed for each extension and each type of problem. As a consequence, the analysis tools have special-purpose efficient algorithms and data-structures suited and optimized only towards their specific problem and extension.

In this work we aim to provide a uniform solution to the analysis of timed systems by designing a generic semi-algorithm to analyze real-time programs which semantically captures a wide range of specification formalisms, including hybrid automata.
We demonstrate that our new method provides solutions to problems which are unsolvable by the current state-of-the-art tools.
We also show that our technique can be extended to solve specific problems like robustness and parameter synthesis.

\paragraph{\bfseries Related work.}

The \emph{trace abstraction refinement} (TAR) technique was  proposed by Heizmann \emph{et al.}~\cite{traceref,HeizmannHP-cav-13}.
Wang \emph{et al.}~\cite{Wang2014} proposed the use of TAR for the analysis of timed automata.
However, their approach is based on the computation of the standard \emph{zones} which comes with usual limitations: it is not applicable to extensions of TA (\eg stopwatch automata) and can only discover predicates that are zones. Moreover, their approach has not been implemented and it is not clear whether it can outperform state-of-the-art techniques \eg as implemented in \uppaal.

Several works have investigated CEGAR techniques in both timed and hybrid settings~\cite{10.1007/3-540-45873-5_6,cegartime,10.1007/978-3-540-78929-1_14,Tiwari2008}. 
The CEGAR technique has also been extended to parameter-synthesis~\cite{10.1007/978-3-540-78929-1_14}.
As proved by  Heizmann \emph{et al.}~\cite{traceref}, such methods are special cases of the TAR framework.

The \emph{IC3}~\cite{10.1007/978-3-642-18275-4_7} verification approach has also been deployed for the verification of hybrid systems~\cite{10.1007/978-3-319-21690-4_41} but relies on a fix-point computation over a combined encoding of the transition-function, rather than a trace-subtraction approach.
IC3 approaches and the likes have also been used for parameter synthesis~\cite{10.1007/3-540-48320-9_14,10.1007/978-3-319-21690-4_41,6679406}.
While similar fundamental techniques are leveraged in these approaches (e.g. \cite{10.1007/3-540-48320-9_14} utilizes Fourier-Motzkin-elimination), we note that our refinement method (TAR) is radically different in nature.
IC3 is an iterative fix-point computation over an up-front and complete encoding of the transition-function.

Since the publication of a preliminary version of this paper~\cite{rprttar}, Kafle\emph{ et al.}~\cite{kafle} have demonstrated a novel method of parameter synthesis for timed systems via Constrained Horn Clauses (CHC).
While their approach shows promising results for the Fischers parameter synthesis examples from \cite{rprttar}, it currently relies on manual translation of a given problem into CHC format, hindering its applicability to large systems.

As mentioned earlier, our technique allows for a unique and logical (predicates) representation of sets of states accross different models (timed, hybrid automata) and problems (reachability, robustness, parameter synthesis), which is in contrast
to state-of-the-art tools such as \uppaal~\cite{uppaal}, \textsc{SpaceEx}~\cite{spaceex}, \textsc{HyTech}~\cite{Henzinger97hytech}, \phaver~\cite{phaver}, \textsc{verifix}~\cite{verifix}, \symrob~\cite{symrob} and \imitator~\cite{imitator} that rely on special-purpose polyhedra libraries to realize their computation.

We propose a new technique which is radically different to previous approaches  and leverages the power of SMT-solvers to discover non-trivial invariants for a large class of real-time systems including the class of hybrid automata.
All the previous analysis techniques compute, reduce and check the state-space either up-front or on-the-fly, leading to the construction of significant parts of the state-space.
In contrast our approach is an abstraction refinement method and the refinements are built by discovering non-trivial program invariants that are not always expressible using zones, or  polyehdra. 
For instance they can express constraints that combine discrete and continuous variables of the system.
This enables us to use our algorithm on non-decidable classes like stopwatch automata, and successfully (\ie the algorithm terminates) check instances of these classes. A simple example is discussed in Section~\ref{sec:example}.

\paragraph{\bfseries Our contribution.}
We propose a variant of the trace abstractions refinement (TAR) technique to solve the reachability problem and the parameter synthesis problem for real-time programs.
Our approach combines  an automata-theoretic framework and  state-of-the-art Satisfiability Modulo Theory (SMT) techniques for discovering program invariants.
We demonstrate on a number of case-studies that this new approach can compute answers to problems unsolvable by special-purpose tools and algorithms in their respective domain.

This paper is an extended version of~\cite{rprttar} in which we first introduced TAR for real-time programs.
In this extended version, we provide a comprehensive introduction illustrated by more examples, extensions of the original algorithms from~\cite{rprttar} and the proofs of theorems and lemmas.

\section{Motivations} \label{sec:example}

\paragraph{\bf \itshape Real-Time Programs.}
\noindent \figref{fig-ex1} is an example of a real-time program $P_1$. 
It is defined by a finite automaton $A_1$ (\figref{fig-ex1}, top) which is the \emph{control flow graph (CFG)} of $P_1$, and
 some \emph{continuous and discrete instructions} (bottom).
The control flow graph $A_1$ accepts the regular language $\lang(A_1) = i.t_0.t_1^\ast.t_2$: the program starts in (control) location $\iota$ and is completed when $\ell_2$ (accepting location) is reached.
The program variables are $x, y, z$ which are real numbers.
This real-time program is the specification of a stopwatch automaton with 2 clocks, $x$ and $z$, and one stopwatch $y$.
The variables  are updated according to the following rules: 
\begin{itemize}
    \item Each edge's label defines a \emph{guard} $g$ (a condition on the variables) for which the edge is enabled, and an \emph{update} $u$ which is an assignment to the variables when the edge is taken.
    For instance the edge $t_1$ can be taken when the valuation of the variable $x$ is $1$ and when it is taken, $x$ is reset. This corresponds to a \emph{discrete} transition of the program. 

    \item Each location is associated with a \emph{rate vector} $r$ that defines the derivatives of the variables. 
    The default derivative for a variable is $1$ (we omit the rates for $x,z$ in the Figure).
    For instance in location $\ell_0$ the derivatives are $(1,0,1)$ (order is $x,y,z$). When the program is in location $\ell_0$ the variables $x,y,z$ increase at a rate defined by their respective derivatives:
    $x,z$ increase by $1$ each time unit, and $y$ is frozen (derivative is $0$). This corresponds to a \emph{continuous} transition of the program.
    
\end{itemize}

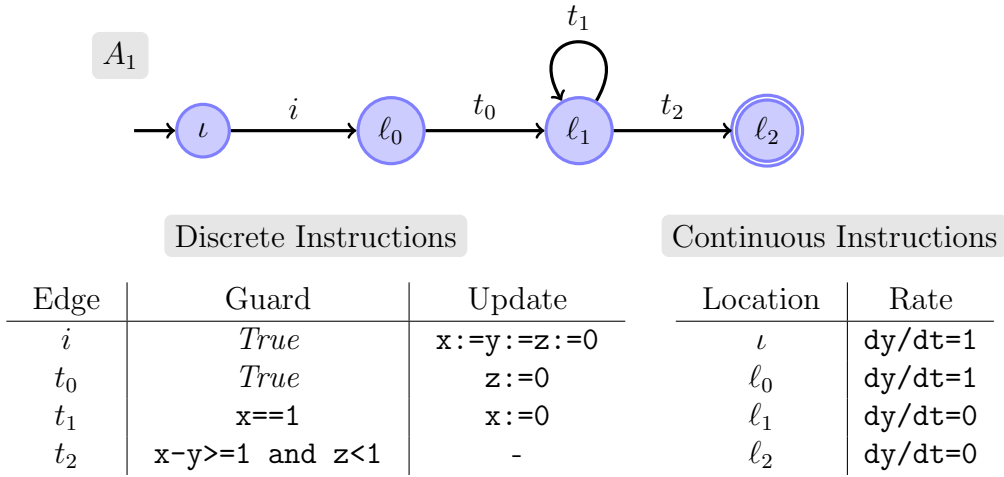
\begin{figure}[thbp]
    \centering
    \begin{tikzpicture}[node distance=2.5cm and 1.5cm,very thick]
    %
        \node[initial left,state,minimum size=0.7cm] (i) {$\iota$};
        \node[state, right of=i] (0) {$\ell_0$};
        \node[state, right of=0] (1) {$\ell_1$};
        \node[state, accepting,right of=1] (2) {$\ell_2$};

        \node[left of=i,yshift=1cm,xshift=1.4cm,fill=Gray,rounded corners=1mm](autname) {$A_1$};
        \node[below of=0,xshift=-1cm,yshift=-.8cm] (transitions) { 
            \begin{tabular}{c|c|c}
                ~Edge~ & Guard & Update  \\\hline
                $i$ & $\true$ & ~\texttt{x:=y:=z:=0}~ \\
                $t_0$ & $\true$ & \texttt{z:=0}  \\
                $t_1$ & \texttt{x==1} & \texttt{x:=0}  \\
                $t_2$ & ~\texttt{x-y>=1 and z<1}~ & - 
            \end{tabular}
        };
        \node[above of=transitions,rounded corners=1mm,fill=Gray,yshift=-.6cm](trans) {Discrete Instructions};

        \node[right of=transitions,xshift=4.4cm] (rates) { 
            \begin{tabular}{c|c}
                ~Location~ &  Rate \\\hline
                $\iota$ &  \texttt{dy/dt=1}\\
                $\ell_0$  & \texttt{dy/dt=1} \\
                $\ell_1$ & \texttt{dy/dt=0} \\
                $\ell_2$ &  \texttt{dy/dt=0}
            \end{tabular}
        };
        \node[above of=rates,rounded corners=1mm,fill=Gray,yshift=-.6cm](rates) {Continuous Instructions};

        \path[->]
        (i) edge node {$i$} (0)
        (0) edge node {$t_0$} (1)
        (1) edge[loop,in=120,out=60,distance=12mm,swap] node {$t_1$} (1)
        (1) edge node {$t_2$} (2)
      ;
    \end{tikzpicture}
    \caption{Real-time program $P_1$: CFG $A_1$ of $P_1$ (top) with the accepting location $\ell_2$ and its instructions (bottom).}
    \vspace*{-.1cm}
    \label{fig-ex1}
\end{figure}
A sequence of program instructions  $w=a_0.a_1.\cdots.a_n \in\langof(A_1)$ defines a (possibly empty) set of \emph{timed words}, $\timed(w)$, of the form $(a_0,\delta_0).\cdots$ $(a_n,\delta_n)$ where $\delta_i \geq 0, i\in[0..n]$ is the time elapsed between two discrete transitions.
For instance, the timed words associated with $i.t_0.t_2$ are of the form $(i,\delta_0).(t_0,\delta_1).(t_2,\delta_2)$, for all $\delta_i \in \setR_{\geq 0}, i\in\{0,1,2\}$ such that the following constraints (predicates that define that each transition can be fired after $\delta_i$ time units) can be satisfied\footnote{We assume the program starts in $\iota$ and all the variables are initially zero.}:
\begin{align}
    & \underbrace{x_0 =  y_0 = z_0 = \delta_0 \wedge \delta_0 \geq 0}_{\text{ Time elapsing $\delta_0$ in $\iota$}} \quad \wedge  \underbrace{\true}_{\text{Guard of $i$}} &   \label{eq-1} \tag{$C_0$} \\
     & \underbrace{x_1 = y_1 = z_1 = 0 + \delta_1  \wedge \delta_1 \geq 0}_{\text{Update of $i$ and time elapsing $\delta_1$ in $\ell_0$}}  \quad \wedge   \underbrace{\true}_{\text{Guard of $t_0$}} &  \label{eq-2} \tag{$C_1$}  \\
     & \underbrace{x_2 = x_1 + \delta_2 \wedge y_2 = y_1 \wedge  z_2 = 0 + \delta_2 \wedge \delta_2 \geq 0}_{\text{Update of $t_0$ and time elapsing $\delta_2$ in $\ell_1$}} \ \wedge \ \underbrace{x_2  - y_2 \geq 1 \wedge z_2  < 1}_{\text{Guard of $t_2$}}    & \label{eq-3} \tag{$C_2$}
\end{align}
These constraints encode the following semantics: $i$ is taken after $\delta_0$ time units and at that time $x, y, z$ are equal to $\delta_0$ and hence
$x_0, y_0, z_0$ are the values of the variables when
location $\ell_0$ is entered.
The program remains in $\ell_0$ for $\delta_1$ time units.
When $t_0$ is taken after $\delta_1$ time units, the values of $x,y,z$ is given by $x_1, y_1, z_1$.
Finally the program remains $\delta_2$ time units in $\ell_1$ and $t_2$ is taken to reach $\ell_2$ which is the end of the program.
It follows that the program can execute $i.t_0.t_2$ (or in other words, $i.t_0.t_2$  is feasible) if and only if we can find $\delta_0, \delta_1, \delta_2$ such that $C_0 \wedge C_1 \wedge C_2$ is satisfiable. Hence the set of timed words associated with $i.t_0.t_2$ is not empty iff $C_0 \wedge C_1 \wedge C_2$ is satisfiable.

\paragraph{\bf \itshape Language Emptiness Problem.}
The \emph{timed language}, $\tlang(P_1)$, accepted by $P_1$ is the set of timed words associated with all the (untimed) words  $w$ accepted by $A_1$ \ie $\tlang(P_1) = \cup_{w \in \lang(A_1)} \tau(w)$.

The \emph{language emptiness problem} is a standard problem in Timed Automata theory~\cite{AlurDill} and is stated as follows for real-time programs: 
\begin{center}\bf
    given a real-time program $P$, is  $\tlang(P)$ empty?
\end{center}
It is known that the emptiness problem is decidable for some classes of real-time programs like Timed Automata~\cite{AlurDill}, but undecidable for more expressive classes like Stopwatch Automata~\cite{Henzinger199894}.
It is usually possible to compute symbolic representations of sets of
\emph{reachable} valuations after a sequence of transitions. However, to compute the
set of reachable valuations we may need to explore an arbitrary and unbounded number of sequences.
Hence only semi-algorithms exist to compute the set of reachable valuations.
For instance, using  \phaver to compute the set of reachable valuations for $P_1$ does not terminate (\tabref{tab-sym-comp}).
To force termination, we can compute an over-approximation of the set of reachable valuations.
Computing an over-approximation is sound (if we declare an over-approximation of a timed language to be empty the timed language is empty) but incomplete \ie
it may result in \emph{false positives} (we declare a timed language non empty whereas it is empty). This is witnessed by the  column ``\uppaal'' in \tabref{tab-sym-comp} where \uppaal over-approximates sets of valuations in the program $P_1$ using DBMs.
After $i.t_0$, the over-approximation is $0 \leq y \leq x \wedge 0 \leq z \leq x$ (this is the smallest DBMs that contains the actual set of valuations reachable after $i.t_0$). This over-approximation intersects the guard  $x - y \geq 1 \wedge z < 1$ of $t_2$ which enables $t_2$. Using this over-approximate set of valuations we would declare that $\ell_2$ is reachable in $P_1$ but this is an artifact of the over-approximation.\footnote{\uppaal terminates with the result ``the language \textbf{may} not be empty''.}
Neither \uppaal nor \phaver can prove that  $\tlang(P_1) = \varnothing$.

\begin{table}[thbtp]\centering
    \begin{tabular}{||l||l|c||}\hline\hline
        ~Sequence~ & \multicolumn{1}{c|}{\phaver} & \multicolumn{1}{c||}{\uppaal}  \\\hline
        $i.t_0$ &  $z = x - y \wedge 0 \leq z \leq x$ & $0 \leq y \leq  x \wedge 0 \leq z \leq x$ \\\hline
        $i.t_0.t_1$ &  $z = x - y + 1 \wedge 0 \leq x \leq z \leq x + 1$ & $0 \leq z - x \leq 1 \wedge 0 \leq y$ \\\hline
        $i.t_0.(t_1)^2$ &  $z = x - y + 2 \wedge 0 \leq x   \leq z -1 \leq x + 1$ &  $1 \leq z - x \leq 2 \wedge 0 \leq y$ \\\hline
        $i.t_0.(t_1)^3$ &  $z = x - y + 3 \wedge 0 \leq x   \leq z -2 \leq x + 1$ & $2 \leq z - x \leq 3 \wedge 0 \leq y$ \\\hline
        \ldots & \ldots & \ldots \\\hline
        $i.t_0.(t_1)^k$ &  ~$z = x - y + k \wedge 0 \leq x   \leq z - k + 1 \leq x + 1$~ & ~$k - 1 \leq z - x \leq k \wedge 0 \leq y$~ \\\hline
        \ldots & \ldots & \ldots \\\hline\hline
    \end{tabular}
    \smallskip
    \caption{Symbolic representation of reachable states after a sequence of instructions.   \uppaal concludes that $\tlang(A_1) \neq \varnothing$  due to the over-approximation using DBMs. \phaver does not terminate.}
    \label{tab-sym-comp}
\end{table}
\paragraph{\bf \itshape Trace Abstraction Refinement for Real-Time Programs.}
The technique we introduce can discover \emph{arbitrary abstractions and invariants} that enable us to prove  $\tlang(P_1) = \varnothing$. Our method is a version of the \emph{trace abstraction refinement} (TAR) technique introduced in~\cite{traceref} and is depicted in \figref{fig-tar}.

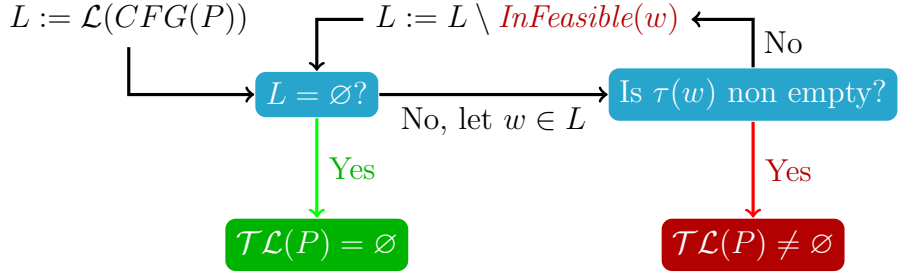
\begin{figure}[hbtp]\centering

    \begin{tikzpicture}[scale=2,node distance=2cm and 2cm, very thick, bend angle=20,bend angle=10]
    
        \node[module](0,0) (init) {$L = \varnothing$?};
        
        \node[module,right of=init, xshift=3.8cm] (checkfeas) {Is $\tau(w)$ non empty?};
       
         \node[below of=init,green] (nobug) { $\tlang(P) = \varnothing$};
       
         \node[below of=checkfeas,red] (bug) {$\tlang(P) \neq \varnothing$};
       
         \node[left of=init,xshift=-.5cm,yshift=1cm] (start) {$L := \lang(CFG(P))$};
       
       \path[->] (init) edge[draw=green] node {\green{Yes}} (nobug)
                        edge[swap] node[yshift=-.0cm] {No, let $w \in L$} (checkfeas)
                (checkfeas) edge[draw=red] node {\red{Yes}} (bug)
          ;
       
        \node[above of=init,yshift=-1cm,xshift=2.8cm] (refine) 
           {$L := L \setminus \red{\textit{InFeasible}(w)}$};
        \draw[->] (checkfeas) |- (refine.east);
        \draw[->] (refine.west) -| (init.north);
        \draw[->] (start) |- (init.west);

        \node[above of=checkfeas, yshift=-1.3cm, xshift=0.4cm] () {No};
       
       \end{tikzpicture}
       \caption{Trace Abstraction Refinement Loop for Real-Time Programs}
       \label{fig-tar}
\end{figure}

\smallskip 
\noindent Let us first introduce how the trace abstraction refinement algorithm (\figref{fig-tar}) operates on a real-time program $P$:
\begin{enumerate}
    \item the algorithm starts using the control flow graph of $P$, $CFG(P)$, and initially $L = \lang(CFG(P))$.

    \item if $L = \varnothing$ then $\tlang(P)$ is empty and the algorithm terminates (green block).

    \item otherwise, there is $w \in L$. We  check whether $\tau(w)$ is empty or not:
    
    \begin{itemize}
        \item If it is not empty then $\tlang(P)$ is not empty and the algorithm terminates (red block). 
        
        \item Otherwise, we can find\footnote{How this language is built is defined in Section~\ref{sec:tar}.} a regular language over the alphabet of $CFG(P)$, $\textit{InFeasible}(w)$, that satisfies: 1) $w \in \textit{InFeasible}(w)$ and 2) $\forall v \in \textit{InFeasible}(w)$, $\tlang(v) = \varnothing$. In the next iteration of the algorithm, we look for a candidate trace in $L \setminus \textit{InFeasible}(w)$, \ie we refine the trace abstraction $L$ by subtracting $\textit{InFeasible}(w)$ from it.
    \end{itemize}
    
\end{enumerate}
Assume the algorithm terminates after $k$ iterations.
In this case we were able to build a finite number of regular languages 
$L_1 = \textit{InFeasible}(w_1), 
L_2 = \textit{InFeasible}(w_2), \cdots, 
L_k = \textit{InFeasible}(w_k)$ such that 
$\forall 1 \leq i \leq k, \tlang(L_i) = \varnothing$.
If we terminate with $\tlang(P) = \varnothing$ then $\lang(CFG(P))  \subseteq \cup_{i = 1}^k L_i$.
 Otherwise if we terminate with $\tau(w_{k+1}) \neq \varnothing$ we found a witness trace $w_{k+1} \in \lang(CFG(P))  \setminus \cup_{i = 1}^k L_i$ such that $\tau(w_{k+1}) \neq \varnothing$ \ie a feasible timed trace.

\paragraph{\bf \itshape Example 1: Stopwatch Automaton.}

\noindent We illustrate the algorithm using our program $P_1$:

\begin{itemize}
    \item we initially let $L=\lang(CFG(P_1))$. Since $w_1= i.t_0.t_2 \in \lang(CFG(P_1))$ and thus $w_1\in L$ the check $L=\emptyset$ fails. 
    We therefore check whether $\timed(w_1) = \varnothing$ which can be done by encoding the  corresponding set of
     timed traces as described by Equations~\eqref{eq-1}--\eqref{eq-3} and then check whether 
     $C_0 \wedge C_1 \wedge C_2$ is satisfiable (\eg using an SMT-solver and the theory of Linear Real Arithmetic). $C_0 \wedge C_1 \wedge C_2$
     is not satisfiable and this establishes $\timed(w_1) = \varnothing$.

    \item from the proof that $C_0 \wedge C_1 \wedge C_2$ is not satisfiable, we
     can obtain  an \emph{inductive interpolant} that comprises of two predicates $I_0, I_1$ -- one for each conjunction -- over the variables $x,y,z$.
     An example of an inductive interpolant\footnote{This is the pair returned by Z3 for $C_0 \wedge C_1 \wedge C_2$.} is $I_0 = x \leq y$ and  $I_1 =  x - y \leq z$.
     These predicates are \emph{invariants} of any timed word of the untimed word $w_1$, and can be used
     to annotate the sequence of transitions $w_1$ with pre- and post-conditions (Equation~\ref{eq-1a}), which are Hoare triples of the form $\{C\} \ a \ \{D\}$:
     \begin{equation}
         \{ \true \}  \quad i \quad \{I_0\} \quad t_0 \quad \{I_1\} \quad t_2 \quad \{\false\} \label{eq-1a} 
     \end{equation}
     A triple $\{C\} \ a \ \{D\}$ is \emph{valid} if whenever we start in a state $s$ satisfying $C$, and execute instruction $a$, the resulting new state $s'$ is in $D$.  $\{C\} \ a \ \{ \false\}$ means that no state exists after executing $a$ from $C$, \ie the trace $a$ is infeasible.
     The inductiveness of the interpolants is due to the fact that each triple 
     $\{C\} \ a \ \{D\}$ in the sequence \eqref{eq-1a} is a valid Hoare triple.
     Hoare triples (and validity) generalise to sequences of instructions $\sigma$ in the form $\{C\} \ \sigma \ \{D\}$.

     Because we can also prove that $\{I_1\}\ (t_1)^\ast\ \{I_1\}$ is a valid Hoare triple, if we combine this fact with Equation~\ref{eq-1a} we obtain a regular set of traces annotated with pre/post-conditions  as per Equation~\ref{eq-1b}.
     \begin{equation}
        \{ \true \} \quad i \quad \{I_0\} \quad t_0 \quad  \pmb{\{I_1\}} \quad \pmb{(t_1)^\ast} \quad \pmb{\{I_1\} } \quad t_2 \quad \{\false\} \label{eq-1b}
    \end{equation}
    This implies that the regular set of traces $i.t_0.(t_1)^\ast.t_2$ does not have any associated timed traces:
     for each word $w \in i.t_0.(t_1)^\ast.t_2$, $\tau(w) = \varnothing$ and as
      $\lang(CFG(P_1)) \subseteq i.t_0.(t_1)^\ast.t_2$ we can conclude that $\tlang(A_1) = \varnothing$.

\end{itemize}

\paragraph{\bf \itshape Example 2: Extended Timed Automaton.}
The following example (\figref{fig-ex2}) illustrates how extensions of timed automata
 with constraints that mix discrete and real variables can be analyzed.
The real-time program $P_2$ (\figref{fig-ex2}) is given by the CFG (left) and the instructions\footnote{The rates table is omitted as all the variables are clocks with rate $1$.} (right): it specifies a timed automaton with $2$ clocks $x,y$ (real variables) and one integer variable $i$. 
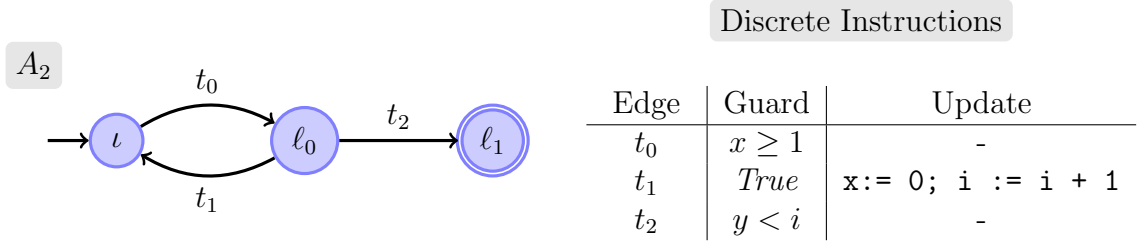
\begin{figure}[htbp]
    \centering
    \begin{tikzpicture}[node distance=2.5cm and 1.5cm,very thick]
    %
        \node[initial left,state,minimum size=0.7cm] (i) {$\iota$};
        \node[state, right of=i] (0) {$\ell_0$};
        \node[state, accepting,right of=0] (2) {$\ell_1$};

        \node[left of=i,yshift=1cm,xshift=1.4cm,fill=Gray,rounded corners=1mm](autname) {$A_2$};
        \node[right of=2,xshift=2.4cm,yshift=-0.3cm] (transitions) { 
            \begin{tabular}{c|c|c}
                ~Edge~ & Guard & Update  \\\hline
                $t_0$ & $x \geq 1$ & -  \\
                $t_1$ & $\true$ & \texttt{x:= 0; i := i + 1}  \\
                $t_2$ & $y < i$ & -  
            \end{tabular}
        };
        \node[above of=transitions,rounded corners=1mm,fill=Gray,yshift=-.6cm](trans) {Discrete Instructions};
        \path[->]
        (i) edge[bend left] node {$t_0$} (0)
        (0) edge[bend left] node {$t_1$} (i)
        (0) edge node {$t_2$} (1)
      ;
    \end{tikzpicture}
    \caption{Real-time program $P_2$: CFG $A_2$ of $P_2$ (left) with accepting location $\ell_1$ and its instructions (right).}
    \label{fig-ex2}
\end{figure}
This is an extended version of timed automata as the constraint $y < i$ mixes integer and real variables (clocks) and this is not permitted in the standard definition of timed automata.
Initially all the variables are set to $0$.
The objective is to prove that location $\ell_1$ is unreachable and thus that $\tlang(P_2) = \varnothing$.
Note that \uppaal does allow this specification but is unable to prove that $\ell_1$ is unreachable because $i$ is unbounded.

Our method is able to discover invariants that mix integer and real variables and can prove that $\ell_1$ is unreachable as follows:
\begin{enumerate}
    \item the first iteration of the TAR algorithm starts with $L = \lang(CFG(P_2))$.
    The check $L = \varnothing$ is negative as $w_1 = t_0.t_2 \in L$. However every timed word in $\tau(w_1)$ must satisfy the following constraints that correspond to taking $t_0$ and then $t_2$:
    \begin{align}
          & \underbrace{x_0 = y_0 =  \delta_0 \wedge \delta_0 \geq 0 \wedge i_0 = 0}_{\text{Time elapsing $\delta_0$ in $\iota$}} \, \wedge \underbrace{x_0 \geq 1}_{\text{Guard of $t_0$}} \label{eq2-1} & \tag{$C'_0$} \\
        & \underbrace{x_1 =  x_0 + \delta_1 \wedge y_1 = y_0 + \delta_1 \wedge i_1 = i_0 \wedge \delta_1 \geq 0}_{\text{Update of $t_0$ and time elapsing $\delta_1$ in $\ell_0$}}  \,  \wedge \underbrace{y_1 < i_1}_{\text{Guard of $t_2$}} &  \label{eq2-2} \tag{$C'_1$}
   \end{align}
   $C'_0 \wedge C'_1$ is not satisfiable and hence $\tlang(t_0.t_2) = \varnothing$ and thus we can safely remove $w_1$ from $L$.
   We can extract interpolants from the proof of unsatisfiability of $C'_0 \wedge C'_1$ and we establish the following sequence of valid Hoare triples:
   \begin{equation}
    \{ x = y = i = 0 \}  \quad t_0 \quad \{ x = y \wedge x \geq i \} \quad t_2 \quad \{\false\} \label{eq2-1a} 
\end{equation}

   \item  the second iteration of the TAR algorithm starts with an  updated $L = \lang(CFG(P_2)) \setminus \{ w_1\}$. 
   Again $L$  is not empty and for instance $w_2 = t_0.t_1.t_0.t_2$ is in $L$. The encoding for checking the emptiness of $\tau(w_2)$ is:
   \begin{align}
    &  \underbrace{x_0 = y_0 =  \delta_0 \wedge \delta_0 \geq 0 \wedge i_0 = 0}_{\text{Time elapsing $\delta_0$ in $\iota$}} \, \wedge \underbrace{x_0 \geq 1}_{\text{Guard of $t_0$}}   & \label{eq3-1} \tag{$C''_0$} \\
    & \underbrace{x_1 =  x_0 + \delta_1 \wedge y_1 = y_0 + \delta_1 \wedge i_1 = i_0  \wedge \delta_1 \geq 0}_{\text{Update of $t_0$ and time elspsing $\delta_1$ in $\ell_0$}} \wedge \underbrace{\true}_{\text{Guard of $t_1$}}  &  \label{eq3-2} \tag{$C''_1$}\\
    &  \underbrace{x_2 = 0 + \delta_2 \wedge y_2 = y_1 + \delta_2 \wedge  i_2 = i_1 + 1 \wedge \delta_2 \geq 0}_{\text{Time elapsing $\delta_2$ in $\iota$}} \wedge \underbrace{x_2 \geq 1}_{\text{Guard of $t_0$}}  &  \label{eq3-3} \tag{$C''_2$}\\
    & \underbrace{x_3 =  x_2 + \delta_3 \wedge y_3 = y_2 + \delta_3 \wedge  i_3 = i_2  \wedge \delta_3 \geq 0}_{\text{Time elapsing $\delta_3$ in $\ell_0$}} \wedge \underbrace{y_3 < i_3}_{\text{Guard of $t_2$}} &  \label{eq3-4} \tag{$C''_3$}
\end{align}
$C''_0 \wedge C''_1 \wedge C''_2 \wedge C''_3$ is unsatisfiable and hence  $\tlang(t_0.t_1.t_0.t_2) = \varnothing$. 
We can extract interpolants from the proof of unsatisfiability  and we establish the following sequence of valid Hoare triples.
   \begin{equation}
    \{ x = y = i = 0 \}  \quad t_0 \quad \{ y \geq i \} \quad t_1.t_0 \quad  \{ y \geq i \} \quad t_2 \quad \{\false\} \label{eq3-1a}
\end{equation}
As can be seen as $\{ y \geq i \} \ t_1.t_0 \  \{ y \geq i \}$ holds we can generalize this sequence to an arbitrary number of iterations of $t_0.t_1$:
\begin{equation}
    \{ x = y = i = 0 \}  \quad t_0 \quad \pmb{\{ y \geq i \}} \quad \pmb{(t_1.t_0)^+} \quad  \pmb{\{ y \geq i \}} \quad t_2 \quad \{\false\} \label{eq3-1b}
\end{equation}
which entails that $\tlang(t_0.(t_1.t_0)^+.t_2) = \varnothing$. This implies that we can remove $t_0.(t_1.t_0)^+.t_2$ from $L$.
\item observe that $L=\emptyset$ in the next iteration of TAR  as $\lang(CFG(P_2))\setminus( \{ t_0.t_2 \} \cup t_0.(t_1.t_0)^+.t_2)=\emptyset$ given that $\lang(CFG(P_2))=t_0.(t_1.t_0)^\ast.t_2$.
We have thus proved that $\tlang(P_2) = \varnothing$ as any word of instructions in $\lang(CFG(P_2))$ induces an infeasible trace and the algorithm terminates.
\end{enumerate}
In the rest of the paper, we provide a formal development of the methods we have introduced so far.

\section{Real-Time Programs}

Our approach is general enough and applicable to a wide range of timed systems
called \emph{real-time programs}.
  As an example, timed, stopwatch, hybrid automata and time Petri nets are special cases of real-time programs.

In this section we formally define \emph{real-time programs}. Real-time programs specify the control flow of \emph{instructions}, just as standard imperative programs do. The instructions can update \emph{variables} by assigning new values to them. Each instruction has a semantics and together with the control flow this precisely defines the semantics of real-time programs.

\subsection{Notations}
A finite automaton over an alphabet $\Sigma$ is a tuple $\automata=(\states,\iota,\alphabet,$ $\tfunc,\accept)$ where $Q$ is a finite set of locations s.t. $\iota \in \states$ is the initial location, $\alphabet$ is a finite alphabet of actions, $\tfunc \subseteq (\states \times \alphabet \times \states)$ is a finite transition relation, $\accept \subseteq Q$ is the set of \emph{accepting} locations.
A word $w = \lbl_0 . \lbl_1.  \cdots.  \lbl_n$ is a finite sequence of letters from $\Sigma$; we let $w[i] = \lbl_i$ be the $i$-th letter of $w$,
$|w|$ be the length of $w$ which is $n + 1$.
Let $\epsilon$ be the empty word and $|\epsilon|=0$, and let $\Sigma^\ast$ be the set of finite words over $\Sigma$. 
The \emph{language}, $\lang(\automata)$,  accepted by $\automata$ is defined in the usual manner as the set of words that can lead to $\accept$ from $\iota$.

Let $V$ be a finite set of real-valued variables. 
A \emph{valuation} is a function $\val: \rvars \rightarrow \setR$.
The set of valuations is $[\rvars \rightarrow \setR]$.

We denote by $\bexpr(\rvars)$ the set of \emph{constraints} (or Boolean predicates) over $\rvars$ and
given $\varphi \in \bexpr(\rvars)$, we let $\vars(\varphi)$ be the set of unconstrained variables in $\varphi$.
Given a valuation, we let the truth value of a constraint (Boolean predicate) $\csystem$ be denoted by $\varphi(\val) \in \{ \true, \false\}$, 
and write $\val \models \csystem$ when  $\csystem(\val) = \true$ and let $\semof{\varphi} = \{ \nu \mid \val \models \csystem\}$.

An \emph{update} $\updt \subseteq [\rvars \rightarrow \setR] \times [\rvars \rightarrow \setR]$ is a binary relation over valuations.
Given an update $\updt$ and a set of valuations $\vals$, we let $\updt(\vals) = \{ \val' \mid \exists \val \in \vals \text{ and } (\val,\val') \in \updt \}$.
We let $\calU(\rvars)$ be the set of updates on the variables in $\rvars$.

Similar to the update relation, we define a \emph{rate} function $\rates:\rvars\rightarrow\setR$
 (rates can be negative), \ie a function from a variable to a real number\footnote{We can allow rates to be arbitrary terms but in this paper we restrict to deterministic rates or bounded intervals.}.
A rate is then a vector 
 $\rates\in \setR^\rvars$.
Given a valuation $\val$ and a timestep $\delta \in \setR_{\geq 0}$ the valuation $\val + (\rates,\delta)$ is defined by: $(\val + (\rates,\delta))(v) = \val(v) + \rates(v) \times \delta$ for $v \in \rvars$.

\subsection{Real-Time Instructions}
Let $\calI = \bexpr(\rvars) \times \calU(\rvars) \times \calR(\rvars)$ be a countable set of instructions -- and intentionally also the alphabet of the CFG. 
Each $\lbl \in \calI$ is a tuple $(\textit{guard, update, rates})$ denoted by $(\guard_\lbl, \updt_\lbl, \rates_\lbl)$.
Let $\val: \rvars \rightarrow \setR$ and $\val'  : \rvars \rightarrow \setR$ be two valuations.
For each pair $(\lbl,\delta) \in \calI \times \setR_{\geq 0}$
we define the following transition relation $\xrightarrow{~ \lbl, \delta ~}$:
\[
\val \xrightarrow{~ \lbl, \delta ~} \val' \iff
\begin{cases}
	1.\quad \val \models \guard_\lbl \text{ (guard of $\alpha$ is satisfied in $\nu$)}, \\
	2.\quad \exists \val'' \text{ s.t. } (\val, \val'') \in \updt_\lbl
	\text{ (discrete update allowed by $\alpha$) and } \\
	3. \quad \val' = \val'' + (\rates_\lbl,\delta) 
	\text{ (continuous update as defined by $\alpha$).}
\end{cases}
\]
The semantics of $\lbl \in \calI$ is a mapping $\fut{\lbl} : [V \rightarrow \setR] \rightarrow 2^{ [V \rightarrow \setR]}$ and for $\nu \in [V \rightarrow \setR]$
\begin{equation}
	\fut{\lbl}(\val) =  \{ \val' \, | \, \exists \delta \geq 0, \val\xrightarrow{~ \lbl, \delta ~} \val' \}\mathpunct. 
   \end{equation}
It follows that:
\begin{fact}\label{fact-1}
	$\exists \delta \geq 0, \val \xrightarrow{~ \lbl, \delta ~} \val' \iff \nu' \in \fut{\lbl}(\val)$. 
\end{fact}
This mapping can be straightforwardly extended to sets of valuations $K \subseteq [V \rightarrow \setR]$ as follows:
\begin{equation}
 \fut{\lbl}(K) =  \underset{\val \in K}{\bigcup} \fut{\lbl}(\val)\mathpunct.
\end{equation}

\subsection{Post Operator}
Let $K$ be a set of valuations and $w \in \calI^\ast$.
 We inductively define the \emph{(strongest) post operator} $\post(K,w)$ as follows:
 \begin{align*}
	 \post(K, \epsilon) & = K \\
	 \post(K, \lbl.w) & = \post(\fut{\lbl}(K),w)
 \end{align*}
The post operator extends to logical constraints $\varphi \in \bexpr(\rvars)$ by defining $\post(\varphi,w) = \post(\semof{\varphi},w)$.
 In the sequel, we assume that, when $\varphi \in \bexpr(\rvars)$,
 then $\fut{\lbl}(\semof{\varphi})$ is also definable as a constraint in $\bexpr(\rvars)$. This inductively implies that $\post(\varphi,w)$ can also be expressed as a constraint in $\bexpr(\rvars)$ for sequences of instructions $w \in \calI^\ast$.

\subsection{Timed Words and Feasible Words}
A \emph{timed word} (over alphabet $\calI$) is a fini\-te sequence  $\word=(\lbl_0,\delta_0).(\lbl_1,\delta_1). \cdots.(\lbl_n,\delta_n)$ such that for each $0 \leq i \leq n$, $\delta_i \in \setR_{\geq 0}$ and
$\lbl_i \in \calI$.
The timed word $\word$ is \emph{feasible}  if and only if there exists a set of  valuations $\{\val_0,\dots,\val_{n+1}\} \subseteq [\rvars \rightarrow \setR]$ such that:
	\[
		\val_0 \xrightarrow{~\lbl_0, \delta_0 ~} \val_1 \xrightarrow{~\lbl_1, \delta_1 ~} \val_2 \quad \cdots \quad \val_n \xrightarrow{~\lbl_n, \delta_n ~} \val_{n+1} \mathpunct.
	\]
We let $\unt(\word) = \lbl_0.\lbl_1.\cdots.\lbl_n$ be the \emph{untimed} version of $\word$. We extend the notion \emph{feasible} to an untimed word $w \in \calI^\ast$: $w$  is feasible iff
$w = \unt(\word)$ for some feasible timed word $\word$.

\begin{lemma}\label{lemma-2}
	An untimed word $w \in \calI^\ast$ is \emph{feasible}  iff $\post(\true,w) \neq \false$.
\end{lemma}
\begin{proof}
We prove this Lemma by induction on the length of $w$.
The induction hypothesis is: 
\[
	\val_0 \xrightarrow{~\lbl_0, \delta_0 ~} \val_1 \xrightarrow{~\lbl_1, \delta_1 ~} \val_2 \quad \cdots \quad \val_n \xrightarrow{~\lbl_n, \delta_n ~} \val_{n+1} \iff \nu_{n+1} \in \post(\{\nu_0\},\lbl_0.\lbl_1.\cdots.\lbl_n)
\]
which is enough to prove the Lemma.

\noindent{\it Base step.} If $w = \epsilon$, then $\post(\{ \nu_0 \},\epsilon) = \{\nu_0\}$.

\noindent{\it Inductive step.} 
Assume $\val_0 \xrightarrow{~\lbl_0, \delta_0 ~} \val_1 \xrightarrow{~\lbl_1, \delta_1 ~} \val_2 \quad \cdots \quad \val_n \xrightarrow{~\lbl_n, \delta_n ~} \val_{n+1} \xrightarrow{~\lbl_{n+1}, \delta_{n+1} } \val_{n+2}$.
By induction hypothesis, $\val_{n+1} \in \post(\{\nu_0\}, \lbl_0.\lbl_1.\cdots.\lbl_n)$, and $\val_{n+2} \in \fut{\lbl_{n+1}}(\val_{n+1})$.
By definition of $\post$ this implies that $\val_{n+2} \in \post(\{\val_0\}, \lbl_0.\lbl_1.\cdots.\lbl_n.\lbl_{n+1})$.
\qed
\end{proof}

\smallskip

\subsection{Real-Time Programs}

The specification of a real-time program decouples the \emph{control} (\eg for Timed Automata, the locations) and the \emph{data} (the clocks or integer variables).
A \emph{real-time program} is a pair $\prog=(A_\prog,\semof{\cdot})$ where $A_\prog$ is a finite automaton $A_\prog=(Q,\iota,\calI, \Delta, F)$ over the alphabet\footnote{$\calI$ can be infinite but we require the control-flow graph $\Delta$ (transition relation) of $A_\prog$ to be finite.} $\calI$, $\Delta$ defines the control-flow graph of the program and $\semof{\cdot}$ provides the semantics of each instruction.

\noindent A timed word $\word$ is \emph{accepted} by  $\prog$ if and only if:
\begin{enumerate}
	\item $\unt(\word)$ is accepted by $A_\prog$ and, 
	\item $\word$ is feasible.
\end{enumerate}
The \emph{timed language}, $\tlang(\prog)$,  of a real-time program $\prog$ is the set of timed words accepted by $\prog$, \ie $\word \in \tlang(\prog)$ if and only if $\unt(\word) \in \lang(A_\prog)$ and $\word$ is feasible.
\begin{remark}
	 We do not assume any particular values initially for the variables of a real-time program (the variables that appear in $I$).
	 This is reflected by the definition of \emph{feasibility} that only requires the existence of valuations without containing the initial one $\nu_0$.
	When specifying a real-time program, initial values can be explicitly set by regular instructions at the beginning of the program.
	This is similar to standard programs where the first instructions can set the values of some variables.
\end{remark}

\subsection{Timed Language Emptiness Problem}
The \emph{(timed) language emptiness prob\-lem} asks the following:
\begin{quote}
	Given a real-time program $\prog$, is $\tlang(\prog)$ empty?
\end{quote}
\begin{theorem}
	$\tlang(\prog) \neq \varnothing$ iff
    $\exists w \in \lang(A_\prog)$ such that $\post(\true,w) \not\subseteq \false$.
\end{theorem}
\begin{proof}
	$\tlang(\prog) \neq \varnothing$ iff there exists a feasible timed word $\word$ such that $\unt(\word)$ is accepted by $A_\prog$. This is equivalent to
	the existence of a feasible word $w \in \lang(A_\prog)$, and
	 by Lemma~\ref{lemma-2}, feasibility of $w$ is equivalent to $\post(\true,w) \not\subseteq \false$.\qed
\end{proof}

\subsection{Useful Classes of Real-Time Programs}

\emph{Timed Automata} are a special case of real-time programs.
The variables are called clocks.
$\bexpr(\rvars)$ is restricted to constraints on individual clocks or \emph{difference constraints} generated by the grammar:
\begin{equation}\label{eq-grammar-ta}
b_1, b_2::= \true \mid \false \mid x - y \Join k \mid x \Join k \mid b_1\wedge b_2
\end{equation}
where $x, y\in V$, $k \in \setQ_{\geq 0}$ and $\Join\,\in\{<,\leq,=,\geq,>\}$\footnote{
While difference constraints are strictly disallowed in most definitions of Timed Automata, the method we propose
retain its properties regardless of their presence.}.
We note that wlog. we omit \emph{location invariants} as for the language emptiness problem, these can be implemented as guards.
An update in $\mu \in \calU(\rvars)$ is defined by a set of clocks to be \emph{reset}.
Each pair $(\nu,\nu') \in \mu$ is such that $\nu'(x)=\nu(x)$ or $\nu'(x)=0$ for each $x \in \rvars$.
The valid rates are fixed to 1, and thus $\calR(\rvars) = \{1\}^\rvars$. 

\smallskip
\emph{Stopwatch Automata} can also be defined as a special case of real-time programs.
As defined in~\cite{stopwatches}, Stopwatch Automata
are Timed Automata extended with \emph{stopwatches} which are clocks that can be stopped. $\bexpr(\rvars)$ and $\calU(\rvars)$ are the same as for Timed Automata but the set of valid rates is defined by the functions of the form $\calR(\rvars) = \{0,1\}^\rvars$ (the clock rates can be either $0$ or $1$). An example of a Stopwatch Automaton is given by the timed system $\automata_1$ in \figref{fig-ex1}.

As there exists syntactic translations (preserving timed languages or reachability) that map hybrid automata to stopwatch automata~\cite{stopwatches}, and translations that map time Petri nets~\cite{DBLP:conf/formats/BerardCHLR05,cassez-jss-06} and extensions~\cite{tpn-13,BJJJMS:TCS:13} thereof to timed automata, it follows that time Petri nets and hybrid automata are also special cases of real-time programs.
This shows that the method we present in the next section is applicable to a wide range of timed systems.
 What is remarkable as well, is that it is not restricted to timed systems that have a finite number of discrete states but can also accommodate infinite discrete state spaces. For example, the real-time program $P_2$ in \figref{fig-ex2}, page~\pageref{fig-ex2} has two clocks $x$ and $y$ and an unbounded integer variable $i$.
Even though $i$ is unbounded, our technique discovers the loop invariant $y \geq i$ of the $\iota$ and $\ell_0$ locations --
an invariant is over a real-time clock $y$ and the integer variable $i$.
 It allows us to prove that $\tlang(P_2) = \varnothing$ as the guard of $t_2$ never can be satisfied ($y<i$).

\section{Trace Abstraction Refinement for Real-Time Programs}\label{sec:tar}
In this section we give a formal description of a semi-algorithm to solve the language emptiness problem for real-time programs. The semi-algorithm is a version of the \emph{refinement of trace abstractions} (TAR) approach~\cite{traceref} for timed systems.

\subsection{Refinement of Trace Abstraction for Real-Time Programs}
We have already introduced our algorithm in~\figref{fig-tar}, page~\pageref{fig-tar}.
We now give a precise formulation of the TAR semi-algorithm for real-time programs,
in~\algref{algo-1}.
It is essentially the same as the  semi-algorithm as introduced in~\cite{traceref} --  we  therefore omit theorems of completeness and soundness as these will be equivalent to the theorems in~\cite{traceref} and are proved in the exact same manner.

\begin{algorithm}
    \small
    \SetAlFnt{\small}
    \SetAlFnt{\small}
    \SetAlCapFnt{\small}
    \SetAlCapNameFnt{\small}
    \SetAlgoVlined
    \SetNoFillComment
    
    \SetKwInOut{Input}{Input}
    \SetKwInOut{Output}{Result}
    \SetKwInOut{Data}{Var}
    
    \newcommand{\myrcomment}[1]{\hfill \small \textcolor{gray}{\textit{/* #1 */}}\par}
    \newcommand{\mycomment}[1]{\small \textcolor{gray}{\textit{/* #1 */}} \hfill\par}
    \Input{A real-time program $\prog=(A_\prog, \semof{\cdot})$.} 
    \Output{$(\true, -)$ if $\tlang(\prog) = \varnothing$, and otherwise $(\false, w)$ if $\tlang(\prog)\neq \varnothing$ with $w \in \lang(A_P)$ and $\post(\true,w) \not\subseteq \false$ -- or non-termination.}
    \Data{$R$: a regular language, initially $R = \varnothing$.\\
      $w$: a word in $ \lang(A_P)$, initially $w = \epsilon$.\\
      $T$: A finite automaton, initially empty. 
    }
    \While{$\lang(A_P)  \not\subseteq R$}{
      Let $w \in \lang(A_P) \setminus R$;\\
      \eIf{
        $\post(\true,w) \not\subseteq \false $}{
        \tcc{$w$ is feasible and $w$ is a counter-example}
        \Return $(\false, w)$\;  
      }{%
        \tcc{$w$ is infeasible, compute an interpolant automaton based on $w$}
        Let $T = \textit{ITA}(w)$\;
        \tcc{Add $T$ to refinement and continue}
        Let $R := R \cup \lang(T)$\;
      }  
    }
    \Return $(\true, - )$\;
    \caption{RTTAR -- Trace Abstraction Refinement for Real-Time Programs}
    \label{algo-1}
    \end{algorithm}

The input to the semi-algorithm \emph{TAR-RT} is a real-time program $\prog=(A_\prog, \semof{\cdot})$.
An invariant of the semi-algorithm is that the refinement $R$, which is subtracted to the initial set of traces, is either empty or containing infeasible traces only.
In the coarsets, initial abstraction, all the words $\lang(A_\prog)$ are potentially feasible. 
In each iteration of the algorithm, we then chip away infeasible behaviour (via the set $R$) of $A_\prog$, making the set difference $\lang(A_\prog)\setminus R$ move closer to the set of feasible traces, thereby shrinking the overapproximation of feasible traces ($\lang(A_\prog)\setminus R$). 

Initially the refinement $R$ is the empty set.
The semi-algorithm works as follows:
\begin{description}
    \item[Step~1] line~1, check whether all the (untimed) traces in $\lang(A_\prog)$ are in $R$. If this is the case, $\tlang(P)$ is empty and the semi-algorithm terminates (line~8). Otherwise (line~2), there is a sequence $w \in \lang(A_\prog) \setminus R$, goto Step~2;

    \item[Step~2] if $w$ is feasible (line~3) \ie there is a feasible timed word $\word$ such that $\unt(\word)=w$, then $\word \in \tlang(P)$ and $\tlang(P) \neq \varnothing$ and the semi-algorithm terminates (line~4).
    Otherwise $w$ is not feasible, goto Step~3;

    \item[Step~3] $w$ is infeasible and given the reason for infeasibility we can construct (line~6) a finite \emph{interpolant automaton}, $\ia(w)$, that accepts $w$ and other words that are infeasible for the same reason. How $\ia(w)$ is computed is addressed in the sequel.
    The automaton $\ia(w)$ is added (line~7) to the previous refinement $R$ and the semi-algorithm starts a new round at Step~1 (line~1).

\end{description}
In the next paragraphs we explain the main steps of the algorithms: how to check feasibility of a sequence of instructions and how to build $\ia(w)$.

\subsection{Checking Feasibility}
Given a arbitrary word  $w \in \calI^\ast$, we can check whether $w$ is feasible by encoding the side-effects of each instruction in  $w$ using linear arithmetic as demonstrated in Examples~1 and~2.

We now define a function $\enc$ for constructing such a constraint-system characterizing the feasibility of a given trace.
We first show how to encode the side-effects and feasibility of a single instruction
$\alpha \in \calI$.
Recall that $\alpha=(\gamma, \mu, \rho)$ where the three components are respectively the guard, the update, and the rates.
Assume that the variables\footnote{The union of the variables in $\gamma, \mu, \rho$.} in $\alpha$ are $X = \{x_1, x_2, \cdots, x_k\}$.
We can define the semantics of $\alpha$ using the standard unprimed\footnote{$\overline{x}$ denotes the vector of variables $\{x_1, x_2, \cdots, x_k\}$.} and primed variables ($X'$).
We assume that the guard and the updates can be defined by predicates and write  $\alpha=(\csystem(\overline{x}),\updt(\overline{x}, \overline{x}'),\rates(\overline{x}))$ with:
\begin{itemize}
    \item $\csystem(\overline{x}) \in \bexpr(X)$ is the guard of the instruction,
    \item $\updt(\overline{x}, \overline{x}')$ a set of constraints in $\bexpr(X \cup X')$,
    \item $\rates: X \rightarrow \setQ$ defines the rates of the variables.
\end{itemize}
The effect of $\alpha$ from a valuation  $\overline{x}''$, which is composed of 1) discrete step if the guard is true followed by the updates leading to a new valuation $\overline{x}'$, and 2) continuous step \ie time elapsing $\delta$, leading to a new valuation $\overline{x}$, can be encoded as follows: 
\begin{equation}\label{eq-enco-inst}
    \enc(\alpha, \overline{x}'', \overline{x}', \overline{x}, \delta)   =   \csystem(\overline{x}'') \wedge \updt(\overline{x}'', \overline{x}') \wedge \overline{x} = \overline{x}' + (\rates, \delta) \wedge \delta \geq 0  
\end{equation}
Let $K(\overline{x})$  be a set of valuations that can be defined as constraint in $\bexpr(X)$.
It follows that $\fut{\alpha}(K(\overline{x}))$ is defined by:
\begin{equation}
    \exists \delta , \overline{x}'',\overline{x}' \text{ such that } K(\overline{x}'') \wedge \enc(\alpha, \overline{x}'', \overline{x}', \overline{x}, \delta)  \label{eq-post-enc}
\end{equation}
In other terms, $\fut{\alpha}(K(\overline{x}))$ is not empty iff $K(\overline{x}'') \wedge \enc(\alpha, \overline{x}'', \overline{x}', \overline{x}, \delta)$  is \emph{satisfiable}.

\smallskip

We can now define the encoding of a sequence of instructions $w = \lbl_0.\lbl_1. \cdots . \lbl_n \in \calI^\ast$.
Given a set of variables $W$, we define the corresponding set of super-scripted variables $W^k = \{ w^j, w \in W, 0 \leq j \leq k\}$. 
Instead of using $x, x', x''$ we use super-scripted variables $\overline{x}^k$ (and $\overline{y}^k$ for the intermediate variables $x'$) to encode the side-effect of each instruction in the trace:

\[
   \enc(w) =  \bigwedge_{i = 0}^n \enc(\lbl_i, \overline{x}^i, \overline{y}^i, \overline{x}^{i+1}, \delta^i) 
\]

It is straighgforward to prove that the function
$\enc:\calI^\ast\rightarrow\bexpr(X^{n+1} \cup  Y^n \cup \{ \delta \}^n )$ constructs a constraint-system
characterizing exactly the feasibility of a word $w$:
\begin{fact}\label{lem:encoding}
For each $ w \in \calI^\ast$,  $\post(\true,w) \not\subseteq \false$ iff $\enc(w)$ is satisfiable.
\end{fact}
If the terms we build are in a logic supported by SMT-solvers (\eg Linear Real Arithmetic) we can automatically check satisfiability.
If $\enc(w)$ is satisfiable we can even collect some \emph{model} which provides witness values for the $\delta_k$.
Otherwise, if $\enc(w)$ is unsatisfiable, there are some options to collect \emph{some reasons for unsatisfiability}: unsat cores or interpolants. The latter is discussed in the next section.

\smallskip
An example of an encoding for the real-time program $P_1$ (\figref{fig-ex1}) and the sequence $w_1 = i.t_0.t_2$ is given by the predicates in Equation  \eqref{eq-1}--\eqref{eq-3}.
Hence the sequence $w_1 = i.t_0.t_2$ is feasible iff $\enc(w_1) = C_0 \wedge C_1 \wedge C_2$ is satisfiable.
Using a SMT-solver, \eg with Z3, we can confirm that $\enc(w_1)$ is unsatisfiable. The interpolating\footnote{The interpolating feature of Z3 has been phased out from version 4.6.x. However, there are alternative techniques to obtain inductive interpolants \eg using unsat cores~\cite{DBLP:conf/sigsoft/DietschHMNP17}.} solver Z3 can also generate a sequence of interpolants, $I_0 = x \leq y$ and  $I_1 =  x - y \leq z$, that provide a general reason for unsatisfiability and satisfy:
\begin{equation*}
    \{ \true \} \quad i \quad \{I_0\} \quad t_0 \quad  \{I_1\} \quad t_2 \quad \{\false\}\mathpunct.
\end{equation*}
We can use the interpolants to build interpolant automata as described in the next section.

\subsection{Construction of Interpolant Automata}\label{sec:construction}

\subsubsection{Inductive Interpolant}
When it is determined that a trace $w$ is infeasible, we can easily discard such a single trace and continue searching for a different one.
However, the power of the TAR method is to generalize the infeasibility of a single trace $w$ into a family (regular set) of traces.
This regular set of infeasible traces is computed from \emph{a reason of infeasibility} of $w$ and is formally specified by an \emph{interpolant automaton}, $\ia(w)$.
The reason for infeasibility itself can be the predicates obtained by computing strongest post-conditions or weakest-preconditions or anything in between but it must be an \emph{inductive interpolant}\footnote{Strongest post-conditions and weakest pre-conditions can provide inductive interpolants}.

\begin{figure}[t]
    \centering
\begin{tikzpicture}[node distance=1.5cm and 2.5cm,very thick]
        \small
        \node[initial left,state,rectangle] (in) {$\true$};
        \node[state, right of=in, above of=in,yshift=-1.0cm] (s2) {$I_0$};
	\node[state, right of=in, below of=in,yshift=1.0cm] (s4) {$I'_0$};

        \node[state, right of=s4] (s5) {$I'_1$};
        \node[state, right of=s5] (s6) {$I'_2$};
        \node[state, right of=s6] (s7) {$I'_3$};

        \node[state, rectangle,right of=s7, above of=s7,yshift=-1.0cm, accepting] (s0) {$\false$};

	\node[state, right of=s5, above of=s6,yshift=-0.5cm] (s3) {$I'_1$};
	\path[->]
        (s2) edge node {$t_0$} (s3)
        (s3) edge node {$t_2$} (s0)
        (s4) edge[swap] node {$t_0$} (s5)
        (s5) edge[bend left] node {$t_1$} (s6)
	(s6) edge[bend left] node {$t_0$} (s5)
	(s6) edge[swap] node {$t_0$} (s7)
	(s7) edge[swap] node {$t_2$} (s0)
	(in) edge node {$i$} (s2)
	(in) edge[swap] node {$i$} (s4)
      ;
    \end{tikzpicture}
    \vspace*{-.3cm}
	\caption{Interpolant automaton for $\lang(\ia(w_1)) \cup \lang(\ia(w2))$. }
    \label{fig-mix-int-clock-interpol}
\end{figure}
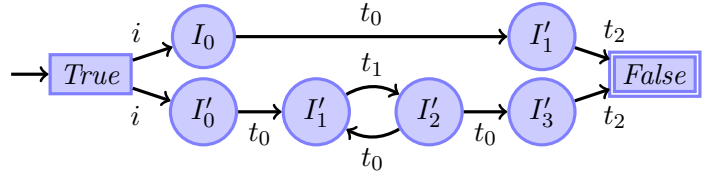
Given a conjunctive formula $f = C_0\wedge\cdots\wedge C_m$,  if $f$ is unsatisfiable,
an \emph{inductive interpolant}  is a sequence of predicates $I_0,\dots,I_{m-1}$ s.t:
\begin{itemize}
    \item $\true \wedge C_0 \implies I_0$,
    \item $I_{m-1} \wedge C_m \implies \false$,
    \item For each $0 \leq n < m-1$, $I_n \wedge C_{n+1} \implies I_{n+1}$, and  the variables in $I_n$ appear in both $C_n$ and $C_{n+1}$ \ie $\vars(I_n) \subseteq \vars(C_n) \cap \vars(C_{n+1})$.

\end{itemize}
If the predicates $C_0, C_1, \cdots, C_m$ encode the side effects of a sequence of instructions $\lbl_0. \lbl_1.\cdots, \lbl_m$, then  one can intuitively think of each interpolant as a \emph{sufficient} condition for infeasibility of the post-fix of the trace and this can be represented by a sequence of valid Hoare triples of the form $\{C\} \ a \ \{D\}$:
\begin{align*}
    \{ \true \} & \quad \lbl_0 \quad \{I_0\} \quad \lbl_1 \quad \{I_1\} \quad \cdots \quad \{ I_{m-1} \} \quad \lbl_m \quad \{\false\}
\end{align*}
Consider the real-time program $P_3$ of \figref{fig-ex2} and the two infeasible untimed words $w_1=i.t_0.t_2$ and $w_2=i.t_0.t_1.t_0.t_2$.
Some inductive interpolants for $w_1$ and $w_2$ can be given by: 
$I_0= y_0 \geq x_0 \wedge (k_0=0)$,
$I_1= y_1 \geq k_1$ for $w_1$ and
$I'_0=y_0 \geq x_0 \wedge k_0 \leq 0$,
$I'_1=y_1 \geq 1 \wedge k_1 \leq 0$,
$I'_2=y_2\geq k_2 + x_2$,
$I'_3=y_3\geq k_3 + 1$ for $w_2$.
From the inductive interpolants one can obtain valid Hoare triples by de-indexing the predicates in the inductive interpolants\footnote{This is a direct result of the encoding function $\enc$. The interpolants can only contain at most one version of each indexed variables.} as shown in Equations~\ref{eq-int1}-\ref{eq-int2}:
\begin{align}
    \{ \true \} & \quad i \quad \{\deidx(I_0)\} \quad t_0  \quad \{ \deidx(I_1)  \} \quad t_2 \quad \{\false\} \label{eq-int1} \\
    \{ \true \} &  \quad i  \quad \{ \deidx(I'_0)  \} \quad t_0 \quad \{\deidx(I'_1) \} \quad t_1 \quad \{ \deidx(I'_2)\} \quad t_0 \quad \{\deidx(I'_3)\} \quad t_2 \quad \{\false\} \label{eq-int2}
\end{align}
where $\deidx(I_k)$ is the same as $I_k$ where each indexed variable $x_j$ replaced by $x$. 
As can be seen in Equation~\ref{eq-int2}, the sequence contains two occurrences of $t_0$: this suggests that a loop occurs in the program, and this loop may be infeasible as well.
Formally, because $\post(\deidx(I'_2),t_0) \subseteq I'_1$, any trace of the form $i.t_0.t_1.(t_0.t_1)^\ast.t_0.t_2$ is infeasible. This enables us to construct an interpolant automaton $\ia(w_2)$  accepting the regular set of infeasible traces $i.t_0.t_1.(t_0.t_1)^\ast.t_0.t_2$.
Overall, because $w_1$ is also infeasible, the \emph{union} of the languages accepted by $\ia(w_2)$ and $\ia(w_1)$ is a set of infeasible traces as defined by the finite automaton in~\figref{fig-mix-int-clock-interpol}.

\smallskip

Given $w$ such that $\enc(w)$ is unsatisfiable we can always find an inductive interpolant: the strongest post-conditions $\post(\true, w[i])$ or (the weakest pre-conditions from $\false$) defines an inductive interpolant.
More generally, we have:
\begin{lemma}\label{lem-int-to-hoaretriple}
    Let $w = \lbl_0.\lbl_1.\cdots.\lbl_m \in \calI^\ast$.
     If $\enc(w) = C_0 \wedge C_1 \wedge \cdots \wedge C_{m}$ is unsatisfiable and  $I_0, \cdots, I_{m-1}$ is an inductive interpolant for $\enc(w)$, the following sequence of Hoare triples
    \[
        \{\true\} \quad \lbl_0 \quad \{ \deidx(I_0)\} \quad \lbl_1 \quad \{ \deidx(I_1)\} \quad \cdots \quad  
            \lbl_{m-1} \quad \{ \deidx(I_{m-1}) \} \quad \lbl_m \quad \{\false\}
        \]
        is valid.
\end{lemma}
\begin{proof}
    The proof follows from the encoding $\enc(w)$ and the fact that each $I_k$ is included in the weakest pre-condition $\textit{wp}(\false,\lbl_{k+1}.\lbl_m)$ which can be proved by induction using the property of inductive interpolants.
    \qed
\end{proof}
\smallskip

\subsubsection{Interpolant Automata}

Let us formalize the interpolant-automata construction.
Let $w =\lbl_0.\lbl_1. \cdots.\lbl_m \in \calI^\ast$, $\enc(w) = C_0 \wedge \dots \wedge C_{m}$ and assume $\post(\true,w) \subseteq \false$ \ie $\enc(w)$ is unsatisfiable (Fact~\ref{lem:encoding}).

Let $I_0,\dots I_{m-1}$ be an inductive interpolant for $C_0 \wedge \dots \wedge C_{m}$. 
We can construct an interpolant automaton for $w$, $\ia(w)=(\states^w,\astate^w_0,\alphabet^w,\tfunc^w,\accept^w)$ as follows:
\begin{itemize}
    \item  $\states^w=\{\true,\false,\deidx(I_0),\cdots,\deidx(I_{m-1})\}$, (note that if two de-indexed interpolants are the same they account for one state only),
    \item $\Sigma^w = \{ \lbl_0,\lbl_1, \cdots,\lbl_m \}$,
    \item $F^w = \{ \false \}$,
    \item $\Delta^w$ satisfies following conditions:
        \begin{enumerate}
            \item $(\true, \lbl_0, \pi(I_0)) \in \Delta^w$,
            \item $(\pi(I_{m-1}), \lbl_m, \false) \in \Delta^w$,
            \item $\forall a \in \Sigma^w, \forall 0 \leq k,j \leq m - 1$, if $\post(\deidx(I_k),a) \subseteq \deidx(I_j)$ then $(\deidx(I_k), a, \deidx(I_{j})) \in \Delta^w$.
        \end{enumerate}
\end{itemize}
Notice that as $\post(\deidx(I_k),\lbl_{k+1}) \subseteq \deidx(I_{k+1})$ the word $w$ itself is accepted by $\ia(w)$ and $\ia(w)$ is never empty.

\begin{theorem}[Interpolant Automata]\label{lem:interpolants}
Let $w$ be an infeasible word over $\prog$, then for all $w'\in\lang(\ia(w))$, $w'$ is infeasible.
\end{theorem}
\begin{proof}
    This proof is essentially the same as the original one in~\cite{traceref}.
    The proof uses rule~3 in the construction of $\ia(w)$: every word accepted by $\ia(w)$
    goes through a sequence of states that form a sequence of valid Hoare triples and end up in $\false$.
    It follows that if $w' \in \ia(w)$, $\post(\true, w') \subseteq \false$. 
     \qed  
\end{proof}

\subsection{Union of Interpolant Automata}\label{sec:union}

In the TAR algorithm we construct interpolant automata at each iteration and the current refinement $R$ is the \emph{union} of the regular languages $\lang(\ia(w_k))$ for each infeasible $w_k$.
The union can be computed using standard automata-theoretic operations.
This assumes that we somehow \emph{forget} the predicates associated with each state of an interpolant automaton.

In this section we introduce a new technique to re-use the information computed in each $\ia(w_k)$ and obtain larger refinements.

\smallskip
Let $A=(Q,q_0, \Sigma, \Delta, F)$ be a finite automaton such that each $q \in Q$ is a predicate in $\varphi(X)$.
We say that $A$ is \emph{sound} if the transition relation $\Delta$ satisfies: $(I,\alpha,J) \in \Delta$ implies that
$\fut{\alpha}(I) \subseteq J$ (or $\post(I,\alpha) \subseteq J$).

Let $R=(Q^R, \{\true\}, \Sigma^R, \Delta^R, \{ \false\})$ be a sound finite automaton that accepts only infeasible traces.
Let $w \in \calI^\ast$ with $w$ infeasible. The automaton $\ia(w)=(Q^w, \{\true\}, \Sigma^w, \Delta^w, \{ \false\})$ built as described in section~\ref{sec:construction} is sound. We can define an \emph{extended union}, $R \uplus \ia(w) = ( Q^R \cup Q^w, \{ \true \}, \Sigma^R \cup \Sigma^w, \Delta^{R \uplus \ia(w)}, \{\false\} )$ of $R$ and $\ia(w)$ with:
\[
    \Delta^{R \uplus \ia(w)} = \{ (p, \alpha, p')  \} \mid \exists (q,\alpha,q') \in \Delta^R \cup \Delta^w \text{ s.t.}  p\subseteq q\text{ and }p'\supseteq q'\}.\label{union2}
\]
It is easy to see that $\lang(R \uplus \ia(w)) \supseteq \lang(R) \cup \lang(\ia(w))$ but also:

\begin{theorem}
    Let $w' \in \lang(R \uplus \ia(w))$. Then $\post(\true,w') \subseteq \false$.
\end{theorem}
\begin{proof}
    Each transition $(p, \alpha, p') $ in $R \uplus \ia(w)$ corresponds to a valid Hoare triple.
    It is either in $\Delta^R$ or $\Delta^w$ and then is valid by construction or it is weaker than an established Hoare triple in $\Delta^R$ or $\Delta^w$.
    \qed
\end{proof}
This theorem allows us to use the $\uplus$ operator in Algorithm~\ref{algo-1} instead of the standard union of regular languages. The advantage is that we re-use already established Hoare triples to build
a larger refinement at each iteration.
\subsection{Feasibility Beyond Timed Automata}

Satisfiability can be checked with an SMT-solver (and decision procedures exist for useful theories).
In the case of timed automata and stopwatch automata, the feasibility of a trace can be encoded in linear arithmetic. The corresponding theory, Linear Real Arithmetic (LRA) is decidable and supported by most SMT-solvers.
It is also possible to encode non-linear constraints (non-linear guards and assignments).
In the latter cases, the SMT-solver may not be able to provide an answer to the SAT problem as
non-linear theories are undecidable. However, we can still build on a semi-decision procedure of the SMT-solver, and if it provides an answer, get the status of a trace (feasible or not).

\subsection{Sufficient Conditions for Termination}
Let us now construct a set of criteria on a real-time program $P=((Q,q_0,\calI, \Delta, F),\semof{\cdot})$ s.t. our proposed method is guaranteed to terminate.
\begin{lemma}{Termination}
The algorithm presented in \figref{fig-tar} terminates if the following three conditions hold.

\begin{enumerate}
	\item For any word $\word\in\calI^\ast$, then $\semof{\word}$ is expressible within a decidable theory (supported by the solver), and\label{term1}
	\item the statespace of $P$ has a finite representation, and\label{term2}
	\item the solver used returns interpolants within the finite  statespace representation.\label{term3}
\end{enumerate}
\end{lemma}
\begin{proof}
First consider the algorithm presented in \figref{fig-tar}, then we can initially state that for each iteration of the loop $R$ grows and thus the NFA representing $R$ ($\automata^R$) must also.
As per the construction presented in Section~\ref{sec:union} we can observe that the transition-function of $\automata^R$ will increase by at least one in each iteration in Step 3. 
If not, the selection of $\word$ between step 1 and step 2 is surely violated or the construction of $\ia$ in step 3 is.

From Conditions \ref{term2} and \ref{term3} we have that the statespace is finitely representable and that these representatives are used by the solver. 
Thus we know that the interpolant automata also has a finite set of states as per the construction of Section~\ref{sec:union}.
Together with the finiteness of the set of instructions, this implies that the transition-function of the interpolant automata must also be finite.
Hence, the algorithm can (at most) introduce a transition between each pair of states with each instruction, but must at least introduce a new one in every iteration.\qed
\end{proof}

As this termination condition relies on the solver, it is heavily dependent on the construction of the solver.
However, if we consider the class of real-time programs captured by Timed Automata, we know that condition \ref{term1} is satisfied (in fact it is Linear Real Arithmetic), condition \ref{term2} is satisfied via the region-graph construction.
This leaves the construction of a solver satisfying condition \ref{term3}, which in turn should be feasible already from condition \ref{term2}, but is practically achievable for TA via extrapolation-techniques and difference bound matrices (or for systems with only non-strict guards; timed-darts or integer representatives).

\section{Parameter Synthesis for Real-Time Programs}\label{sec:synth}

In this section we show how to use the trace abstraction refinement semi-algorithm presented in Section~\ref{sec:tar} to synthesize \emph{good initial values} for some of the program variables, and to check \emph{robustness} of timed automata.
We first define the \emph{Maximal Safe Initial State} problem and then show how to reduce parameter synthesis and robustness to special cases of this problem.

\smallskip

\subsection{\bfseries Maximal Safe Initial Set Problem}
Given a real-time program $\prog$, the objective is to determine a set of \emph{initial valuations} $I \subseteq [\rvars \rightarrow \setR]$
such that, when we start the program in $I$, $\tlang(\prog)$ is empty.

Given a constraint  $I \in \bexpr(\rvars)$, we define the corresponding \emph{assume} instruction by:
$\test(I) = (I, \textit{Id}, \overline{0})$. This instruction leaves all the variables unchanged (discrete update is the identity function and the rate vector is $\overline{0}$) and this acts as a guard only.

Let  $\prog=(Q, q_0, \calI, \Delta, F)$ be a real-time program and $I \in \bexpr(\rvars)$. We define the
real-time program $\test(I).\prog=(Q, \{ \iota \}, \calI \cup \{ \test(I) \}, \Delta \cup \{(\iota,\test(I),q_0)\}, F)$.

\smallskip
The \emph{\bf \itshape maximal safe initial state problem}  asks the following:
\begin{quote}
	\bf Given a real-time program $\prog$, find a maximal $I \in \bexpr(\rvars)$ s.t. $\tlang(\test(I).\prog) = \varnothing$.
\end{quote}
\subsection{\bfseries Semi-Algorithm for the Maximal Safe Initial State Problem}
Let $w \in \lang(\test(I).P)$ be a feasible word. It follows that $\enc(w)$ must be satisfiable.
We can define the set of initial values for which $\enc(w)$ is satisfiable by projecting away all the variables in the encoding $\enc(w)$ except the ones indexed by $0$. Let $I_0 = \exists (\vars(\enc(w)) \setminus X^0) .\enc(w)$ be the resulting (existentially quantified) predicate and $\deidx(I_0)$ be the corresponding constraint on the program variables without indices. We let $\exists_i(w) = \deidx(I_0)$.
It follows that $\exists_i(w)$ is the maximal set of valuations for which $w$ is feasible.
Note that existential quantification for the theory of Linear Real Arithmetic is within the theory via Fourier–Motzkin-elimination -- hence the computation of $\exists_i(w)$ by an SMT-solver only needs support for Linear Real Arithmetic when $P$ encodes a linear hybrid, stopwatch or timed automaton.\footnote{This idea of using Fourier-Motzkin elimination has already been proposed~\cite{10.1007/3-540-48320-9_14} in the context of timed Petri nets.}

\smallskip

The TAR-based semi-algorithm for the maximal safe initial state problem is presented in \figref{fig-rob-emptiness}.
\begin{figure}[ht]
\centering
\begin{tikzpicture}[scale=1,node distance=1.4cm and 1.5cm, very thick, bend angle=20,bend angle=10]
\fontsize{10}{11}\selectfont
 \node[module](0,0) (init) {\textbf{1:} $\tlang(\test(I).P) = \varnothing$?};
   \node[below of=init,xshift=-0cm,green] (nobug) {Maximal safe init is $I$};
  \node[above of=init,xshift=-2.9cm] (start) {$I := \true$};
\node[module,right of=init, xshift=6cm] (step3) {\textbf{2:} $I := I \wedge \neg \exists_i(\unt(\sigma))$};
\coordinate[above of=step3] (o3);
\coordinate[above of=init] (oi);
\path[->] (init.south) edge[draw=green,swap] node {\green{Yes}} (nobug);
\draw[->,draw=red]  (init.north) node[xshift=-0.3cm,yshift=.5cm] {\red{No}} -- (oi) -- node[yshift=0cm] {
				\begin{tabular}{c}
					\red{Let $\word \in \tlang(\test(I).\prog)$}
				\end{tabular}} (o3) -- (step3.north);

 \draw[->] (step3) -- (init);
 \draw[->] (start) |-  ($(init.west)$);
\end{tikzpicture}
\caption{Semi-algorithm $\mathit{SafeInit}$.}
\label{fig-rob-emptiness}
\end{figure}
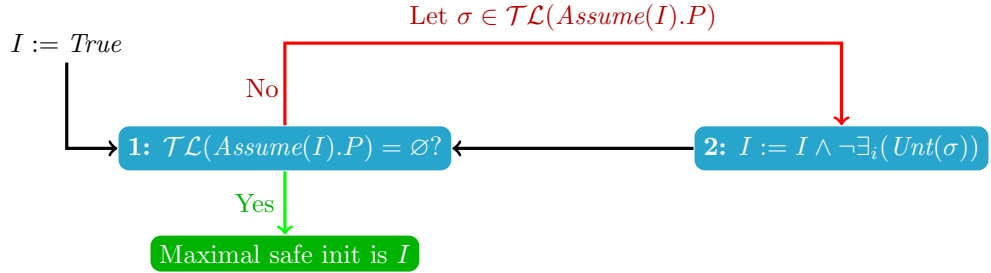
\noindent The semi-algorithm in~\figref{fig-rob-emptiness} works as follows:
\begin{enumerate}
	\item initially $I = \true$
	\item using the semi-algorithm~\ref{algo-1}, check whether $\tlang(\test(I).\prog)$ is empty
	\item if so $\prog$ does not accept any timed word when we start from $\semof{I}$;
	\item Otherwise, there is a witness word $\word \in \tlang(\test(I).\prog)$,
	implying that $I\wedge\enc(\unt(\word))$ is satisfiable.
	It follows that $\exists_i.\enc(\unt(\word))$ cannot be part of the maximal 
	set. It is used to strengthen $I$ and repeating from step~2.
\end{enumerate}

If the semi-algorithm terminates, it computes exactly \textbf{the} maximal set of values for which the system is safe ($I$), captured formally by Theorem~\ref{thm-max-constraint}.
\begin{theorem}\label{thm-max-constraint}
If  the semi-algorithm $\mathit{SafeInit}$ terminates and outputs $I$, then:
\begin{enumerate}
	\item $\tlang(\test(I).\prog) = \varnothing$ and
	\item for any $I'\in \bexpr(\rvars)$, $\tlang(\test(I').\prog) = \varnothing$ implies $I' \subseteq I$.
\end{enumerate}
\end{theorem}
\begin{proof}
	The fact that $\tlang(\test(I).\prog) = \varnothing$ follows from termination.

	The fact that $I$ is maximal is an invariant of the semi-algorithm: at the beginning, $I = \true$ and is clearly maximal. At each iteration, we may subtract a set of valuations $K$ from the previously computed $I$, but these valuations are all such that $\tlang(\test(\nu).P) \neq \varnothing$ for any $\nu\in K$ by definition of existential quantification.

	Hence every time a set of valuations is removed by strengthening $I$ only unsafe initial valuations are removed.
	It follows that if $\textit{safeInit}$ terminates, $I$ is maximal.

	\qed
\end{proof}

\subsection{Parameter Synthesis}

Let  $\prog=(Q, q_0, \calI, \Delta, F)$ be a real-time program over a set of variables $X \cup U$ such that: $\forall u \in U, \forall (g,\updt,\rates) \in \Delta, (\nu,
\nu') \in \updt \implies \nu(u) = \nu'(u)$ and $\rates(u) = 0$. In words, variables in $U$ are constant variables. Note that they can appear in the guard $g$.

\smallskip
The \emph{\bf \itshape parameter synthesis problem}  asks the following:
\begin{quote}
	\bf \hskip-.9em Given a real-time program $\prog$, find a maximal set $I \in \bexpr(U)$ s.t. $\tlang(\test(I).\prog) = \varnothing$.
\end{quote}
The \emph{parameter synthesis problem} is a special case of
the maximal safe initial state problem.
Indeed, solving the maximal safe initial state problem allows us to find the maximal set of parameters such that $\tlang(P) = \varnothing$. 
Let $I$ be a  solution\footnote{For now assume there is a unique maximal solution.} to the maximal safe initial state problem. Then $\exists (\vars(P) \setminus  U).I$ is a maximal set of parameter values such that $\tlang(P) = \varnothing$.

\subsection{Robustness Checking}\label{sec:robust}
Another remarkable feature of our technique is that it can readily be used to check \emph{robustness} of real-time programs and hence timed automata.
In essence, checking robustness amounts to enlarging the guards of a real-time program $P$ by an $\varepsilon > 0$.
The resulting program is $P_\varepsilon$.

\smallskip
The \emph{\bf \itshape robustness problem}  asks the following:
\begin{quote}
	Given a real-time program $\prog$, is there some $\epsilon > 0$, s.t. $\tlang(\prog_\epsilon) = \varnothing$.
\end{quote}

Using our method we can solve the \emph{\bf \itshape robustness synthesis problem} which asks the following:
\begin{quote}
	Given a real-time program $\prog$, find a maximal $\epsilon > 0$, s.t. $\tlang(\prog_\epsilon) = \varnothing$.
\end{quote}
This problem asks for a witness (maximal) value for $\epsilon$.

The robustness synthesis is a special case of the parameter synthesis problem where $\epsilon$ is a parameter of the program $P$.

Note that in our experiments (next section), we assume that $P$ is robust and in this case we can compute a maximal value for $\epsilon$.
Proving that a program is non-robust requires proving \textit{feasibility} of infinite traces for ever decreasing $\epsilon$.
We have developed some techniques (similar to proving termination for standard programs) to do so but this is still under development.

\section{Experiments}\label{sec:experiments}

We have conducted three sets of experiments, each testing the applicability of our proposed method (denoted by \tar) compared to state-of-the-art tools with specialized data-structures and algorithms for the given setting.
All experiments were conducted on AMD EPYC 7551 Processors and limited to 1 hour of computation.
The \tar tool uses the \uppaal parsing-library, but relies on \zthree \cite{z3} for the interpolant computation.
Our experimental setup is available online~\cite{peter_gjol_jensen_2020_3952642}.

\subsection{Verification of Timed and Stopwatch Automata}\label{exp:stopwatch}
The real-time programs, $P_1$ of \figref{fig-ex1} and $P_2$ of \figref{fig-ex2} can be analyzed with our technique.
The analysis (\tar algorithm~\ref{algo-1}) terminates in two iterations for the program $P_1$, a stopwatch automaton. As emphasized in the introduction, neither \uppaal (over-approximation with DBMs) nor \phaver can provide the correct answer to the reachability problem for $P_1$.

To prove that location $2$ is unreachable in program $P_2$ requires to discover an invariant that
mixes integers (discrete part of the state) and clocks (continuous part).
Our technique successfully discovers the program invariants. As a result the refinement depicted in \figref{fig-mix-int-clock-interpol} is constructed
and as it contains $\lang(A_{P_2})$ the refinement algorithm RTTAR terminates and proves that $2$ is not reachable. $A_{P_2}$ can only be analyzed in \uppaal with significant computational effort and bounded integers.

\subsection{Parametric Stopwatch Automata}
We compare the \tar tool to \imitator \cite{imitator} -- the state-of-the-art parameter synthesis tool for reachability\footnote{We compare with the \texttt{EFSynth}-algorithm in the \imitator tool as this yielded the lowest computation time in the two terminating instances.}.
We shall here use the semi-algorithm presented in Section \ref{sec:synth}
For the test-cases we use the gadget presented initially in \figref{fig-ex1}, a few of the test-cases used in \cite{Andre2015}, as well as two modified versions of Fischers Protocol, shown in~\figref{fig:fischer}. 
In the first version we replace the constants in the model with parameters. 
In the second version (marked by robust), we wish to compute an expression, that given an arbitrary upper and lower bound yields the robustness of the system -- in the same style as the experiments presented in Section~\ref{exp:robust}, but here for arbitrary guard-values.

\begin{table}
\centering

\begin{tabular}{|l|cc|}
\hline
\rowcolor{gray!25}
&\multicolumn{1}{c}{\textsc{imitator-2.12}}&\multicolumn{1}{c|}{\textsc{rttar}}\\\hline
\texttt{A1}&DNF&0.08\\\hline
\texttt{Sched2.100.0}&7.16&492.73\\\hline
\texttt{Sched2.50.0}&4.95&273.36\\\hline
\texttt{fischer\_2}&DNF&0.26\\\hline
\texttt{fischer\_2\_robust}&DNF&0.25\\\hline
\texttt{fischer\_4}&DNF&47.96\\\hline
\texttt{fischer\_4\_robust}&DNF&50.26\\\hline

\end{tabular}
	\caption{Results for parameter-synthesis comparing \tar with \imitator. Time is given in seconds. DNF marks that the tool did not complete the computation within an hour.}
\label{tab:synth}
\end{table}
~
\begin{figure}
\centering
\includegraphics[width=5.2cm]{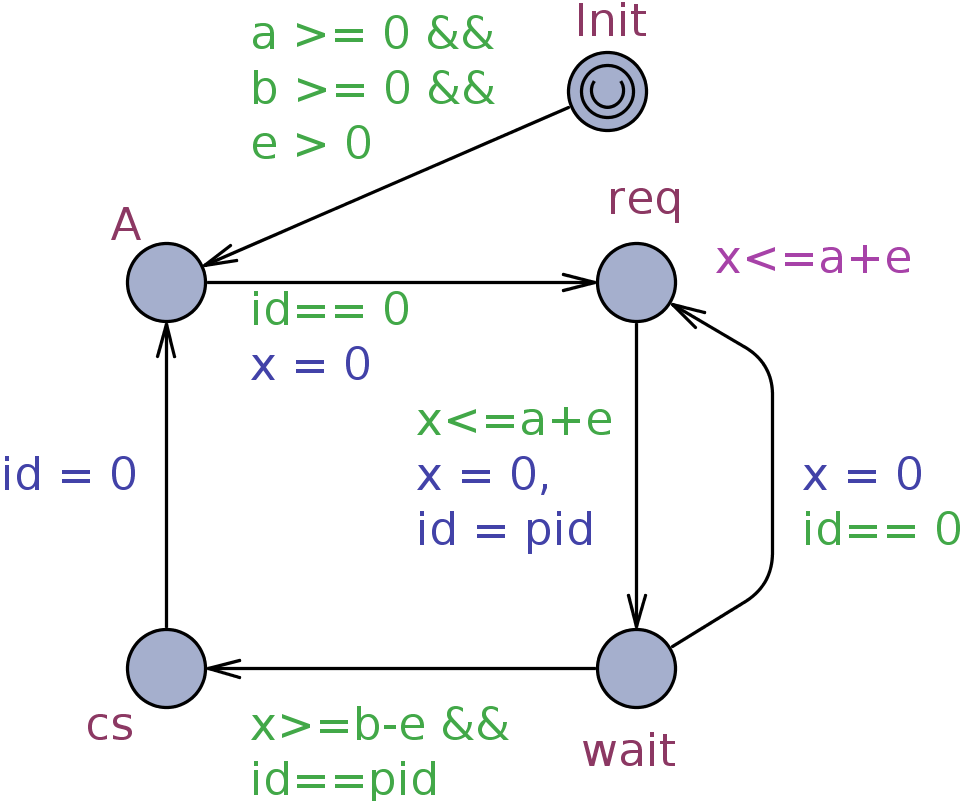}
\captionof{figure}{A \uppaal template for a single process in Fischers Algorithm. The variables \texttt{e}, \texttt{a} and \texttt{b} are parameters for $\epsilon$, lower and upper bounds for clock-values respectively.}
\label{fig:fischer}
\end{figure}

As illustrated by Table~\ref{tab:synth} the performance of \tar is slower than \imitator when \imitator is able to compute the results. On the other hand, when using \imitator to verify our motivating example from \figref{fig-ex1}, we observe that \imitator never terminates, due to the divergence of the polyhedra-computation. This is the effect illustrated in Table~\ref{tab-sym-comp}.

When trying to synthesize the parameters for Fischers algorithm, in all cases, \imitator times out and never computes a result.
For both two and four processes in Fischers algorithm, our tool detects that the system is safe if and only if $a < 0 \vee b < 0 \vee b - a > 0$. Notice that $a < 0 \vee b < 0$ is a trivial constraint preventing the system from doing anything. The constraint $b - a > 0$ is the only useful one. Our technique provides a formal proof that the algorithm is correct for $b -a >0$.

In the same manner, our technique can compute the most general constraint ensuring that Fischers algorithm is robust.
The result of \tar algorithm is that the system is robust iff
$ \epsilon \leq 0 \vee a < 0 \vee b < 0\vee b - a - 2\epsilon > 0$
 -- which for $\epsilon=0$ (modulo the initial non-zero constraint on $\epsilon$) reduces to the constraint-system obtained in the non-robust case.

 \subsection{Robustness of Timed Automata}\label{exp:robust}

To address the robustness problem for a real-time program $P$, we use the semi-algorithm presented in Section~\ref{sec:synth} and  reduce the robustness-checking problem to that of parameter-synthesis.
Notice the delimitation of the input-problems to robust-only instances from Section~\ref{sec:robust}.

\begin{table}
    \centering
	\begin{tabular}{|l
|rr|rr|}
\hline
\rowcolor{gray!25}
&\multicolumn{2}{c|}{\textsc{rttar\_robust}}&\multicolumn{2}{c}{\textsc{symrob}}\\\hline
\texttt{csma\_05}&32.38&\verb!1/3!&0.51&\verb! 1/3!\\\hline
\texttt{csma\_06}&87.55&\verb!1/3!&1.91&\verb! 1/3!\\\hline
\texttt{csma\_07}&294.30&\verb!1/3!&7.37&\verb! 1/3!\\\hline
\texttt{fischer\_04}&17.64&\verb!1/2!&0.19&\verb! 1/2!\\\hline
\texttt{fischer\_05}&102.50&\verb!1/2!&0.77&\verb! 1/2!\\\hline
\texttt{fischer\_06}&519.41&\verb!1/2!&2.83&\verb! 1/2!\\\hline
\texttt{M3}&17.14&$\infty$&\textbf{N/A}&\textbf{N/A}\\\hline
\texttt{M3c}&17.72&$\infty$&3.91&\verb! 250/3!\\\hline
\texttt{a}&3470.95&\verb!1/2!&19.66&\verb! 1/4!\\\hline
\end{tabular}
	\caption{Results for robustness analysis comparing \tar with \symrob. Time is given in seconds. N/A indicates that \symrob  was unable to compute  the robustness for  the given model.}
\label{tab:robustness}
\end{table}
As Table~\ref{tab:robustness} demonstrates, \symrob \cite{symrob} and \tar  do not always agree on the results.
Notably, since the TA \texttt{M3} contains strict guards, \symrob is unable to compute the robustness of it.
Furthermore, \symrob under-approximates $\epsilon$, an artifact of the so-called ``loop-acceleration''-technique and the polyhedra-based algorithm.
This can be observed in the modified model \texttt{M3c}, which is now analyzable by \symrob, but differs in results compared to \tar.
This is the same case with the model denoted \texttt{a}.
We experimented with $\epsilon$-values to confirm that \texttt{M3} is safe for all the values tested -- while \texttt{a} is safe only for values tested respecting $\epsilon<\frac{1}{2}$.
We can also see that our proposed method is significantly slower than the special-purpose algorithms deployed by \symrob, but in contrast to \symrob, it computes the maximal set of good paramaters.

\section{Conclusion}

We have proposed a version of the trace abstraction refinement approach to real-time programs.
We have demonstrated that our semi-algorithm can be used to solve the reachability problem for instances which are not solvable by state-of-the-art analysis tools.

Our algorithms can handle the general class of real-time programs that comprises of classical models for real-time systems including timed automata, stopwatch automata, hybrid automata and time(d) Petri nets.

As demonstrated in Section~\ref{sec:experiments},
our tool is capable of solving instances of reachability problems, robustness, parameter synthesis, that current tools are incapable of handling.

For future work we would like to improve the scalability of the proposed method, utilizing well known techniques such as extrapolations, partial order reduction~\cite{cassez-lpar-2015} and compositional verification~\cite{cassez-fsttcs-2014}.
Another short-term improvement is to use unsat cores to compute interpolant automata as proposed in~\cite{DBLP:conf/sigsoft/DietschHMNP17}.
Furthermore, we would like to extend our approach from reachability to more expressive temporal logics.

\paragraph{Acknowledgments.} 
The research was partially funded by  Innovation Fund Denmark center DiCyPS and ERC Advanced Grant LASSO.
Furthermore, these results was made possible by an external stay partially funded by Otto M\o nsted Fonden.

\bibliographystyle{plain} 
\bibliography{bibliography}

\end{document}